%% file: main.tex
\documentclass[10pt,journal,compsoc]{IEEEtran}

\usepackage{amsmath,amssymb,amsthm}
\usepackage{graphicx}
\usepackage{booktabs}
\usepackage{multirow}
\usepackage{array}
\usepackage{tabularx}
\usepackage{xcolor}
\usepackage{algorithm}
\usepackage{algpseudocode}
\usepackage{xspace}
\usepackage{cite}
\usepackage{hyperref}
\hypersetup{hidelinks}
\usepackage[capitalise,nameinlink]{cleveref}
\usepackage{balance}
\usepackage{microtype}
\usepackage{siunitx}
\usepackage{url}
\usepackage{placeins}
\usepackage{orcidlink}

\newtheorem{theorem}{Theorem}
\newtheorem{proposition}{Proposition}

\newtheorem{remark}{Remark}

\newcommand{\rg}{\textsc{RecurGuard}\xspace}
\newcommand{\qdm}{\textsc{QDM}\xspace}
\newcommand{\tauR}{\tau_{\!R}}
\newcommand{\tauV}{\tau_{\!V}}
\newcommand{\tauP}{\tau_{\!P}}
\newcommand{\taur}{\tau_r}
\newcommand{\simop}{\mathrm{sim}}
\newcommand{\TP}{\mathrm{TP}}
\newcommand{\RR}{\mathrm{RR}}
\newcommand{\VG}{\mathrm{VG}}
\newcommand{\MPD}{\mathrm{MPD}}
\newcommand{\ind}{\mathbf{1}}
\newcommand{\eg}{\emph{e.g.}\xspace}

\newcommand{\heading}[1]{\texorpdfstring{\textbf{#1}}{#1}}
\newcommand{\pointlabel}[1]{\textbf{#1}}

\renewenvironment{abstract}{%
  \par\noindent\textbf{Abstract}\par\smallskip
  \noindent\ignorespaces
}{%
  \par\vspace{0.6\baselineskip}
}
\definecolor{darkblue}{HTML}{1F3864}
\definecolor{alertred}{HTML}{CC0000}
\definecolor{goodgreen}{HTML}{1E5631}

\title{RecurGuard: Runtime Monitoring for Reasoning-Token Consumption Attacks}

\author{Abid~Aziz and Hafsa~Binte~Kibria~\orcidlink{0000-0003-4775-8639}\\
\small Department of Electrical \& Computer Engineering\\
\small Rajshahi University of Engineering \& Technology, Rajshahi 6204, Bangladesh\\
\small 2210019@student.ruet.ac.bd~$|$~hafsabintekibria@ece.ruet.ac.bd
}

\begin{document}
\maketitle

\begin{abstract}
\input{sections/00_abstract}
\end{abstract}

\input{sections/01_introduction}
\input{sections/02_background}
\input{sections/03_attacks}
\input{sections/04_method}
\input{sections/05_theory}
\input{sections/06_setup}
\input{sections/07_results}
\input{sections/08_adaptive}
\input{sections/09_deployment}
\input{sections/10_related}
\input{sections/11_conclusion}

\section*{Funding}
No external funding was received for this research.

\section*{Conflict of Interest}
The authors declare that they have no known competing financial interests or personal relationships that could have appeared to influence the work reported in this paper.

\section*{Data Availability}
The data and code supporting the findings of this study are available at: \url{https://github.com/abidaziz1/recurguard}

\balance
\bibliographystyle{IEEEtran}
\bibliography{refs}

\end{document}

%% file: sections/00_abstract.tex
Reasoning-capable large language models can be induced to spend their
generation budget on injected decoy tasks rather than answering the user's
question, causing denial of service when no final answer is produced and
denial of wallet when excess output tokens are billed. Input-side safety
classifiers often miss these attacks because the injected prompts can appear
syntactically benign. We build \rg{}, a runtime monitor for detecting
reasoning-chain consumption attacks when reasoning traces are exposed by the
model. \rg{} analyses reasoning traces as they are generated and tracks three
signals: recurrence rate, volume growth, and progress toward the user's query.
If all three signals remain anomalous for three consecutive chunks, \rg{}
terminates generation early. We evaluate \rg{} against OverThink and ExtendAttack across open-weight
reasoning models and conduct adaptive stress tests on DS-R1-Qwen-7B. On this
model, \rg{} detects 99\% of OverThink attacks and 92\% of ExtendAttack
instances while maintaining near-zero false positive rates on question
answering, code generation, mathematics, and summarisation. Adaptive evaluation
reveals the limit of the defence: topical attacks retain $11.9\times$
amplification with an approximately 50\% joint miss rate, whereas full semantic
evasion reduces amplification from $22.8\times$ to $2.2\times$. When reasoning
traces are unavailable, \qdm{} provides a post-hoc fallback monitor based on the
final output.

%% file: sections/01_introduction.tex

\section{\heading{Introduction}}
\label{sec:intro}

Reasoning-capable large language models (LLMs) often generate hundreds
or thousands of intermediate reasoning tokens before producing a final
answer~\cite{deepseek2025r1,qwen2025tech}. These extended reasoning
traces improve performance on complex tasks, but they also create a new
attack surface. An attacker who can inject text into the model's context,
through retrieved documents, tool outputs, or user-controlled
inputs~\cite{greshake2023injection}, can introduce a decoy task that
diverts the model's reasoning away from the user's query, an
unbounded-consumption risk~\cite{owasp2024llm} exercised by OverThink
and ExtendAttack~\cite{kumar2025overthink,extendattack2025}. Rather
than solving the intended task, the model spends its generation budget
on the injected decoy.

This behaviour creates two distinct harms. First, the user may receive no
useful answer because the model exhausts its generation budget before
returning to the original task. Second, the excess reasoning incurs
additional inference costs. In our evaluation, successful attacks increase
output-token consumption by more than an order of magnitude and, in some
cases, prevent the model from producing a final answer altogether. These
channels are common in retrieval-augmented generation (RAG), tool-using
agents, and enterprise LLM workflows, where external text can enter the
prompt context.

Three common defence layers are limited in this setting.
\emph{Input-side safety classifiers} such as Llama~Guard~3-1B-INT4~\cite{llamaguard3}
do not detect these attacks because the injected prompt contains a legitimate
question alongside a harmless-looking puzzle or decoding task; no
content-safety category is triggered. \emph{Length-based filters}, including
token-count thresholds and output-length anomaly detectors, catch many
high-amplification attacks but produce unacceptably high false positive rates
on legitimate long-form tasks: a token-budget filter produces 69\% FPR on
HumanEval code generation in our evaluation. \emph{Post-hoc semantic
monitors} flag the final output when it drifts from the user query, but cannot
halt generation early and can confuse legitimate summarisation, where the
output discusses an article rather than the bare instruction, with attack
outputs.

The geometric signal is in the trajectory, not the length alone. Clean
reasoning and attack reasoning move through embedding space differently.
Clean reasoning tends to progress through task-relevant content, with each
chunk introducing new material related to the user's question. Attack
reasoning collapses: after the decoy takes hold, each new chunk resembles the
model's own prior chunks more than it resembles the query. A detector that asks
only ``is the output close to the query?'' is insufficient: it can fire on
summarisation, where the output discusses article content that is semantically
distant from the short instruction. The more useful question is whether each
reasoning chunk is closer to the \emph{query} or to the model's own
\emph{recent reasoning history}. This distinction, which we call
task-conditioned progress, is the foundation of \rg{}.

We propose \textbf{\rg{}}, a lightweight online monitor that processes the
model's exposed reasoning trace chunk by chunk during generation. At each
64-word chunk, \rg{} computes three streaming signals: \emph{recurrence rate}
(RR), measuring how often the current chunk revisits the embedding territory of
recent prior chunks; \emph{volume growth} (VG), measuring whether the embedding
cloud is expanding into new semantic space; and \emph{task-conditioned
progress} (TP), measuring how much more similar the current chunk is to the
user's query than to any prior chunk. A joint alarm fires only when all three
signals indicate collapse for $S{=}3$ consecutive chunks. When the alarm fires,
generation can be halted immediately, saving the remaining output tokens.

For deployment settings where the reasoning trace is inaccessible, such as
closed APIs that return only final outputs, we pair \rg{} with \textbf{\qdm{}}
(Query Drift Monitor). \qdm{} is a lightweight post-hoc companion monitor that
scores final outputs against the query. It is useful when reasoning traces are
hidden, but it does not provide early stopping. We treat \qdm{} as a deployment
fallback rather than the main algorithmic contribution.
\Cref{fig:threat_model} summarises this placement: the injected decoy enters
through text-bearing context, \rg{} monitors exposed reasoning chunks during
generation, and \qdm{} checks the final output when traces are hidden.

\begin{figure*}[t]
\centering
\includegraphics[width=\textwidth]{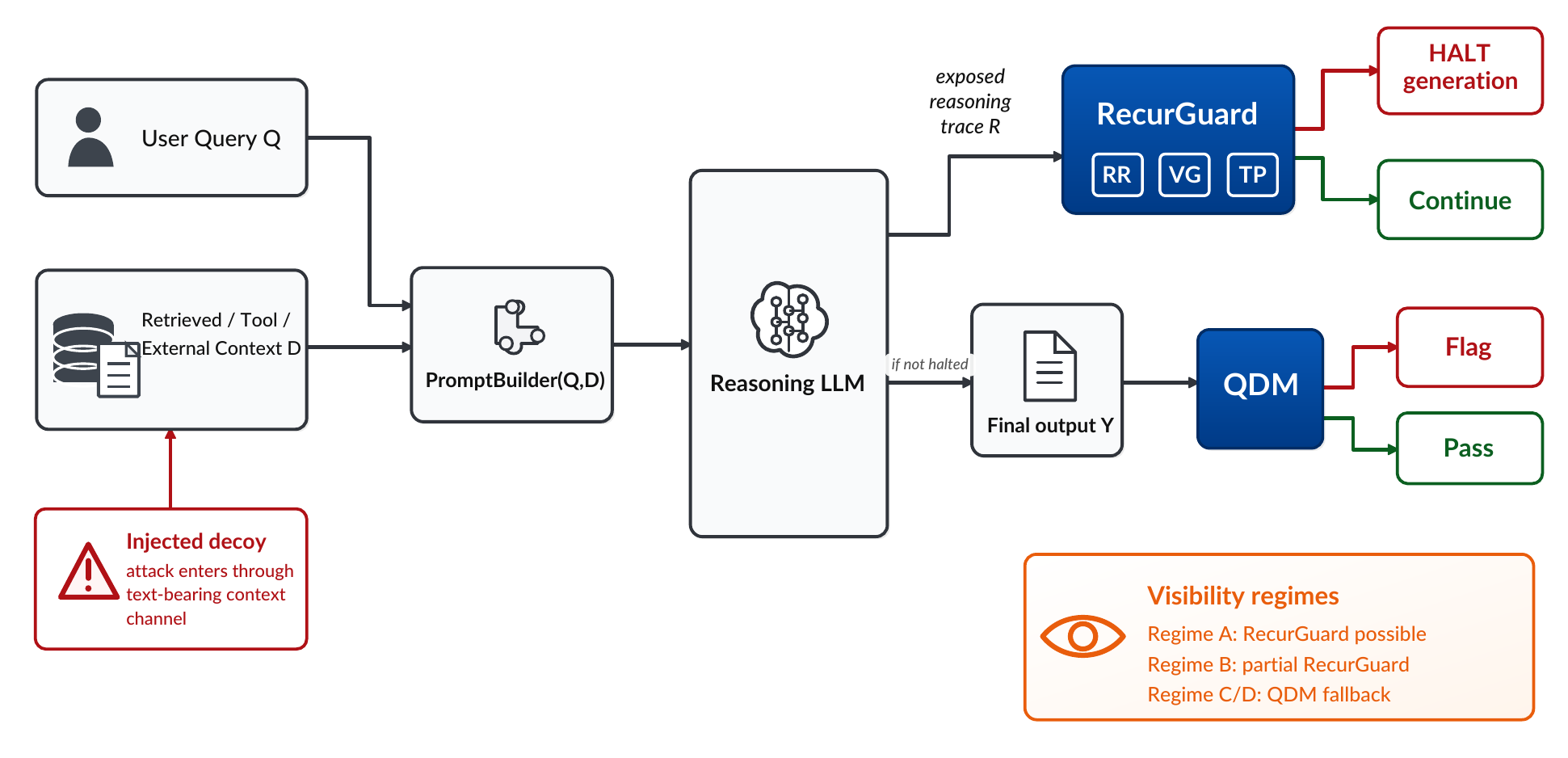}
\caption{Threat model and monitor placement. The attacker injects a
reasoning-consuming decoy through any text-bearing context channel.
\rg{} monitors exposed reasoning traces online and can halt generation;
\qdm{} operates post-hoc on the final output when reasoning traces are
hidden.}
\label{fig:threat_model}
\end{figure*}

On DS-R1-Qwen-7B~\cite{deepseek2025r1}, our primary evaluation model,
\rg{} detects 99\% of OverThink attacks~\cite{kumar2025overthink} with
0\% FPR on the held-out SQuAD~\cite{rajpurkar2016squad} benign set, and
92\% of ExtendAttack instances~\cite{extendattack2025}. Across legitimate
code generation, mathematics, and summarisation tasks, \rg{} maintains
near-zero FPR, while simple length-based baselines produce 69\% and 80\%
FPR on HumanEval~\cite{chen2021humaneval}. Against fully adaptive evasion
(\hyperref[sec:c4_analysis]{C4}), both \rg{} and \qdm{} are evaded, but
the attack's amplification collapses from $22.8\times$ to $2.2\times$,
consistent with the Evasion-Alignment theorem: in our C4 construction,
forcing task alignment causes the model to converge on an answer instead
of looping. The topical CSP attack (\hyperref[sec:c1_analysis]{C1})
remains an open boundary: it retains $11.9\times$ amplification while
achieving an approximately 50\% joint miss rate against both monitors.

This work makes four contributions:
\begin{itemize}
\item \textbf{Threat model.}
We formalise reasoning-chain consumption as a runtime security problem,
introduce defender-visibility regimes, and define token amplification and
liveness failure as the primary harm metrics.

\item \textbf{\rg{}.}
We propose \rg{}, a task-conditioned recurrence-collapse monitor for exposed
reasoning traces, and provide a theoretical analysis explaining its FPR
behaviour and the amplification tradeoff under full evasion.

\item \textbf{Evaluation.}
We evaluate \rg{}, \qdm{}, and four baselines against OverThink, ExtendAttack,
and adaptive stress tests across open-weight and closed-API models, with
Wilson-bounded TPR/FPR reporting and a McNemar paired comparison.

\item \textbf{Security boundary.}
We characterise the adaptive boundary: \hyperref[sec:c1_analysis]{C1} weakens
semantic monitors at $11.9\times$ amplification while remaining detectable by
length baselines; \hyperref[sec:c4_analysis]{C4} achieves full evasion only at
$2.2\times$ amplification.
\end{itemize}

%% file: sections/02_background.tex

\section{\heading{Background and Threat Model}}
\label{sec:threat}

Reasoning-capable large language models may generate an extended
reasoning trace before producing a final answer.
This trace is the attack surface studied here.
A prompt-injection attacker who can insert text into the model's
input context can force the trace to consume the entire output
budget on an irrelevant decoy task, denying the user a useful
answer while incurring the full token cost.
The following subsections define the deployment pipeline,
the attacker's capabilities, the defender's observability, and
the metrics used to measure attack severity.

\subsection{\heading{Reasoning-Chain Consumption Attacks}}
\label{sec:attack_class}

Reasoning models (\eg DeepSeek-R1~\cite{deepseek2025r1},
Qwen3~\cite{qwen2025tech}, Claude with extended
thinking~\cite{anthropic2026thinking}) produce intermediate reasoning text or hidden
reasoning states before answering.
An adversary who can inject text into the prompt context can
include a computationally expensive decoy task, such as a logic
puzzle or obfuscated decoding problem, along with an explicit
instruction to solve it before answering the user's question.
When the model complies, it may spend its entire reasoning budget
on the decoy and produce no useful final answer.

This behaviour causes two distinct harms.
\textbf{Denial of service (DoS):} the user receives no usable
answer, degrading service quality or availability.
\textbf{Denial of wallet (DoW):} the API provider or user is
billed for the full output token count, which can be orders of
magnitude larger than a normal query.
Both harms can arise without harmful content in the output,
so content-safety filters alone may miss them.
This attack class corresponds to
\textbf{OWASP LLM10: Unbounded Consumption}~\cite{owasp2024llm}.

\subsection{\heading{System Setting}}
\label{sec:system}

We study a text-context inference setting that includes direct
prompting and RAG-style deployments. The model receives a combined
prompt constructed from a user query $Q$ and additional context $D$:
\begin{equation}
  P = \mathrm{PromptBuilder}(Q,\,D).
  \label{eq:prompt}
\end{equation}
The model $M$ processes $P$ and produces a reasoning trace $R$
followed by a final answer $A$.
The resulting pipeline is:
\[
  Q,\,D \;\to\; P \;\to\; M \;\to\; R \;\to\; A.
\]
The injected decoy task can enter through any text-bearing
channel: user-controlled input, retrieved document content,
tool call outputs, or external context strings.
The attack does not require modifying the system prompt or the
model itself.

\subsection{\heading{Attacker Capabilities}}
\label{sec:attacker}

The attacker operates in a \emph{black-box prompt-injection}
setting.

\noindent The attacker can:
\begin{itemize}
\item insert or control text appearing in the prompt context $D$;
\item choose a decoy task designed to consume reasoning tokens
  (\eg a Sudoku puzzle or mixed-base ASCII decoding instruction);
\item in the adaptive threat setting, know the general design of
  the deployed detection method.
\end{itemize}

\noindent The attacker cannot:
\begin{itemize}
\item modify model weights or decoding parameters;
\item access hidden activations, logits, or internal model states;
\item modify the system prompt directly;
\item tamper with, disable, or route around the detector;
\item obtain gradients through the detector.
\end{itemize}

We study a black-box prompt-injection attacker, not a white-box
adversary with gradient access to the detector.
The attacker may still craft inputs intended to be missed by the
detector; this is the adaptive evasion setting in \Cref{sec:adaptive}.
Gradient-based optimisation against the sentence encoder is
explicitly outside our threat model and is listed as future work
in \Cref{sec:limitations}.

\subsection{\heading{Defender Visibility}}
\label{sec:visibility}

The defender's ability to apply \rg{} depends on whether the
reasoning trace is accessible.
Table~\ref{tab:visibility} defines four visibility regimes.

\begin{table}[t]
\centering
\caption{Defender visibility regimes.}
\label{tab:visibility}
\footnotesize
\setlength{\tabcolsep}{3pt}
\renewcommand{\arraystretch}{1.05}
\begin{tabularx}{\columnwidth}{@{}cXcc@{}}
\toprule
Regime & Defender observes & \rg{} & \qdm{} \\
\midrule
A & Full reasoning trace   & Yes     & Yes \\
B & Summarised thinking    & Partial & Yes \\
C & Empty/hidden reasoning & No      & Yes \\
D & Final output only      & No      & Yes \\
\bottomrule
\end{tabularx}
\end{table}

Regime~A covers open-weight reasoning models
(\eg DeepSeek-R1, Qwen3) where the trace is fully accessible.
Regime~B covers APIs that expose provider-summarised thinking
blocks (\eg Claude Sonnet~4.5 extended thinking); \rg{} is
applicable but with reduced sensitivity.
Regime~C covers APIs that return empty thinking blocks (\eg
Claude Opus~4.7 in our evaluation); \rg{} is structurally
inapplicable.
Regime~D covers generic closed APIs with no thinking-block
support; only output-side monitoring is available.
Detection statistics for each observed regime are reported in
\Cref{sec:results_closed,sec:deployment}.

\subsection{\heading{Attack Outcomes}}
\label{sec:outcomes}

We characterise an attacked generation by two metrics.

\textbf{Token amplification} measures the increase in output
token count relative to clean queries:
\begin{equation}
  \alpha = \frac{\mathbb{E}[T_{\mathrm{attack}}]}
                {\mathbb{E}[T_{\mathrm{clean}}]},
  \label{eq:amp}
\end{equation}
where $T_{\mathrm{attack}}$ and $T_{\mathrm{clean}}$ are the
output token counts for attacked and clean queries respectively.

\textbf{Liveness failure} captures whether the model produces a
usable answer:
\begin{equation}
  \mathrm{LF}(Y) = \ind\!\left[
    \text{answer tokens} < 5
    \;\lor\;
    T_{\mathrm{attack}} \geq T_{\max}
  \right],
  \label{eq:lf}
\end{equation}
where $T_{\max}$ is the generation budget.
The five-token cutoff treats empty strings, boilerplate fragments, and
truncated answer stubs as no usable answer.
$\mathrm{LF}(Y){=}1$ indicates the model exhausted its budget
without answering.

We distinguish three outcome classes:
\textbf{O1: Normal} denotes $\alpha \approx 1$;
\textbf{O2: Slowdown} denotes $\alpha > 1$ while the model still
answers ($\mathrm{LF}{=}0$);
and \textbf{O3: Liveness failure} denotes $\mathrm{LF}{=}1$.
The DoW threat includes O2 and O3; DoS corresponds to O3.

\subsection{\heading{Scope and Non-Goals}}
\label{sec:nongoals}

We address one specific threat: reasoning-chain consumption via
prompt injection in deployed inference pipelines.

We do not study jailbreaks, harmful-content generation, model
weight poisoning, training-time backdoors, or membership
inference.
These are orthogonal threat classes and our results do not apply
to them.

We do not claim to prevent all long outputs.
Long outputs are legitimate in code generation, mathematical
reasoning, and summarisation.
The goal is to detect reasoning that becomes both
\emph{off-task} and \emph{self-recurrent}, a pattern
characteristic of decoy-driven liveness failure and not
observed in our legitimate long-form evaluations.
The method section formalises this distinction.

%% file: sections/03_attacks.tex

\section{\heading{Attack Mechanisms}}
\label{sec:attacks}

We do not claim the primary attacks as contributions.
We reproduce two published reasoning-chain consumption attacks
and add adaptive variants to test whether semantic monitors
can be evaded.
Both primary attacks are publicly available and have been
evaluated on reasoning LLMs in their original
papers~\cite{kumar2025overthink,extendattack2025}.
Our contribution is the detection framework, not the attacks.

\subsection{\heading{OverThink}}
\label{sec:overthink}

OverThink~\cite{kumar2025overthink} injects an unrelated decoy
task into the prompt context alongside an explicit instruction
to complete it before answering the user's question.
The abstract prompt form is:
\begin{equation*}
  P = \mathrm{concat}\!\left(Q,\; I(D_{\mathrm{decoy}})\right),
  \quad
  I(D)=\texttt{``solve D first.''}
\end{equation*}
In our setup, $D_{\mathrm{decoy}}$ is a Sudoku-style constraint
satisfaction problem.
The decoy tends to trap the model's reasoning process: the model allocates most
of its reasoning budget to the puzzle, either failing to answer
$Q$ or producing a truncated answer after exhausting the token
cap.
The attack does not require any special privilege; it can be
delivered through any text field that reaches the model's
context window, including retrieved documents in a RAG pipeline.

\subsection{\heading{ExtendAttack}}
\label{sec:extendattack}

ExtendAttack~\cite{extendattack2025} takes a different approach.
Rather than inserting an unrelated task, it transforms the
user's own question into an obfuscated representation, such
as a mixed-base ASCII encoding, and instructs the model to
decode it character by character before answering.
The model must reconstruct the original question through
step-by-step decoding, generating substantial reasoning text
before it can address the actual query.

This mechanism differs from OverThink in two ways.
First, the decoy is structurally tied to the original question
rather than being entirely unrelated.
Second, the reasoning pattern is sequential decoding rather
than constraint satisfaction, which produces a different
recurrence signature in embedding space.
Their mechanistic differences motivate testing both as
separate attack families.

\subsection{\heading{Adaptive Variants}}
\label{sec:adaptive_variants}

The two primary attacks rely on decoys that are semantically
distant from the user query.
We design four adaptive stress-test variants to probe whether
semantic monitors can be evaded when the attacker controls the
semantic relationship between the decoy and the query.

\textbf{C1: Topical CSP.}
The decoy borrows vocabulary from the user's question topic.
The goal is to keep the generated reasoning close enough to the
query that query-distance monitors such as \qdm{} hesitate.

\textbf{C2: High-entropy filler.}
This variant avoids a compact repeated decoy and instead spreads
the filler across diverse topics.
The smoke test did not justify scaling it: at $N{=}3$, amplification
stayed at ${\leq}1.5\times$.

\textbf{C3: Structured-diverse filler.}
This variant changes the decoy structure across steps to reduce
repeated chunk similarity.
It also stayed at ${\leq}1.5\times$ amplification in the $N{=}3$
smoke test, so we did not scale it.

\textbf{C4: Task-aligned filler.}
The decoy remains semantically aligned with the user's actual
task, aiming to evade both query-distance and task-conditioned
progress signals simultaneously.

C2 and C3 serve as negative controls: they show that evasion
attempts which destroy recurrence also destroy the attack's
amplification.
We scaled C1 and C4 to $N{=}100$ on the primary model because they
were the only adaptive variants that still produced meaningful
amplification.

\subsection{\heading{Attack Taxonomy}}
\label{sec:taxonomy}

Table~\ref{tab:taxonomy} summarises the six attack configurations
used in this paper.

\begin{table}[t]
\centering
\caption{Attack taxonomy.
``Source'' indicates whether the attack is reproduced from prior
work or is an adaptive stress-test variant.
``Decoy relation'' describes the semantic relationship between
the decoy and the user query.}
\label{tab:taxonomy}
\footnotesize
\setlength{\tabcolsep}{3pt}
\renewcommand{\arraystretch}{1.05}
\begin{tabularx}{\columnwidth}{@{}>{\raggedright\arraybackslash}p{0.27\columnwidth}lXl@{}}
\toprule
Attack & Source & Decoy relation to query & Scaled \\
\midrule
OverThink    & Prior work & Unrelated          & Yes ($N{=}100$) \\
ExtendAttack & Prior work & Encoded query      & Yes ($N{=}50$)  \\
C1: Topical CSP & Stress test & Topically related & Yes ($N{=}100$) \\
C2: High-entropy filler & Stress test & Diverse, unrelated & Smoke ($N{=}3$) \\
C3: Structured-diverse filler & Stress test & Structured diverse & Smoke ($N{=}3$) \\
C4: Task-aligned filler & Stress test & Task-aligned & Yes ($N{=}100$) \\
\bottomrule
\end{tabularx}
\end{table}

These attacks create different geometric patterns in the
reasoning embedding space.
Unrelated decoys (OverThink) produce low query alignment and
high chunk-to-chunk recurrence.
Encoded-query decoys (ExtendAttack) produce moderate query
alignment with sequential recurrence from the decoding process.
Adaptive variants try to preserve query alignment (C4) or
reduce recurrence (C2, C3), while C1 attempts both partially.
\rg{} turns this geometry into streaming detector signals during
generation.

%% file: sections/04_method.tex

\section{\heading{RecurGuard}}
\label{sec:method}

\subsection{\heading{Design Goal}}
\label{sec:design_goal}

\rg{} watches exposed reasoning traces while the model is still
generating.
It looks for the point where the trace stops moving with the user's
task and starts returning to the same decoy reasoning.

The central design choice is to compare each reasoning chunk
against two anchors: the user query and the model's own previous
reasoning.
A chunk is suspicious only when it moves \emph{away from the query}
while becoming \emph{more similar to prior chunks}.
The suspicious case is the conjunction: off-task and self-recurrent
at the same time.
That is the signature of reasoning-chain liveness failure.
Checking either condition alone is insufficient: legitimate
summarisation produces chunks that are off-task but not
self-recurrent, while legitimate analytical reasoning can produce
self-similar but on-task chunks.

\rg{} is paired with \qdm{}, a lightweight output-side monitor for
settings where reasoning traces are inaccessible.
\qdm{} is a deployment companion rather than the main contribution.
\Cref{sec:qdm} gives its design and limitations.

\subsection{\heading{Reasoning Trace Chunking}}
\label{sec:chunking}

Let $R = (w_1, \ldots, w_n)$ be the sequence of words in the
model's exposed reasoning trace, and let $Q$ denote the user query.
We partition $R$ into non-overlapping chunks of $k$ words:
\begin{equation}
  c_t = (w_{(t-1)k+1},\ldots,w_{tk}),
  \qquad t = 1,2,\ldots
  \label{eq:chunk}
\end{equation}
with $k{=}64$.
Each chunk is encoded as a fixed-length embedding:
\begin{equation}
  \mathbf{e}_t = f(c_t) \in \mathbb{R}^d,
  \qquad
  \mathbf{q} = f(Q) \in \mathbb{R}^d,
  \label{eq:embed}
\end{equation}
where $f$ is the SentenceTransformers
\texttt{all-MiniLM-L6-v2} encoder, following the SBERT
sentence-embedding framework~\cite{reimers2019sbert}
and the model card for the released encoder~\cite{sentencetransformers2026minilm}
($d{=}384$), normalised to unit length.
All similarity computations use cosine similarity
$\simop(a,b){=}\langle a,b\rangle$.

\subsection{\heading{Signals}}
\label{sec:signals}

\rg{} computes three scalar signals per chunk using a sliding
window $W_t = \{\max(1,t{-}W),\ldots,t{-}1\}$ of width $W{=}20$.
If the trace is empty because no reasoning text is exposed, \rg{}
returns no alarm and is treated as inapplicable in deployment.
For the first chunk, where no prior chunk exists, we set
$\RR_1{=}0$ and set $\TP_1{=}\simop(\mathbf{e}_1,\mathbf{q})$;
alarm decisions are suppressed by the $c_{\min}$ condition until
enough context is available.
The implementation stores prior 384-dimensional chunk embeddings
for the current generation so TP can compare each new chunk against
earlier chunks; this memory cost is small relative to model
inference in the evaluated runs.

\medskip\noindent\textbf{RR: Recurrence Rate.}
The fraction of window chunks whose similarity to $\mathbf{e}_t$
exceeds a hard inner threshold $\taur$:
\begin{equation}
  \RR_t =
  \frac{1}{|W_t|}
  \sum_{i \in W_t}
  \ind\!\left[\simop(\mathbf{e}_t,\mathbf{e}_i) > \taur\right],
  \qquad \taur{=}0.70.
  \label{eq:rr}
\end{equation}
$\RR_t{=}0$ under diverse reasoning; $\RR_t \to 1$ when each new
chunk revisits the same semantic region.

\medskip\noindent\textbf{VG: Volume Growth.}
The change in mean pairwise distance (MPD) of the embedding cloud
when the current chunk is added:
\begin{align}
  \MPD(S) &=
  \frac{2}{|S|(|S|-1)}
  \sum_{\substack{i<j \\ i,j \in S}}
  \!\left\|\mathbf{e}_i-\mathbf{e}_j\right\|_2,
  \label{eq:mpd}\\
  \VG_t &= \MPD(W_t \cup \{t\}) - \MPD(W_t).
  \label{eq:vg}
\end{align}
For $|S|<2$, we set $\MPD(S)=0$; the implementation assigns early
VG values a non-alarming convention and suppresses alarms through
the $t \geq c_{\min}$ condition.
Positive $\VG_t$ indicates the cloud is expanding (novel content);
negative or near-zero $\VG_t$ indicates semantic collapse
(the reasoning is no longer covering new ground).

\medskip\noindent\textbf{TP: Task-Conditioned Progress.}
\begin{equation}
  \TP_t =
  \simop(\mathbf{e}_t, \mathbf{q})
  \;-\;
  \max_{i<t}\,\simop(\mathbf{e}_t,\mathbf{e}_i).
  \label{eq:tp}
\end{equation}
$\TP_t > 0$ when the current chunk is closer to the query than to
any prior chunk (productive progress).
$\TP_t \ll 0$ when the chunk is far from the query but similar to
prior chunks (recurrence collapse).
TP is the dominant detection signal; the ablation in
\Cref{sec:ablation} shows that TP alone achieves 99\% TPR
and 0\% FPR.

\subsection{\heading{Sustained Alarm Rule}}
\label{sec:alarm}

The per-chunk alarm indicator is:
\begin{equation}
  A_t = \ind\!\left[
    \RR_t > \tauR \;\land\;
    \VG_t < \tauV \;\land\;
    \TP_t < \tauP \;\land\;
    t \geq c_{\min}
  \right],
  \label{eq:alarm}
\end{equation}
with calibrated thresholds
$\tauR{=}0.20$, $\tauV{=}0.01$, $\tauP{=}{-0.738}$,
and minimum pre-alarm chunk count $c_{\min}{=}3$.
An attack is declared when $A_t{=}1$ for $S{=}3$ consecutive
chunks; the counter resets to zero on any $A_t{=}0$.
The sustained requirement prevents isolated anomalous chunks, which
occur in normal reasoning, from triggering false alarms.
Algorithm~\ref{alg:recurguard} gives the full procedure.

\begin{algorithm}[t]
\caption{\rg{}: Task-Conditioned Recurrence Collapse Detection}
\label{alg:recurguard}
\begin{algorithmic}[1]
\Require Streaming reasoning text $R$, query $Q$,
         thresholds $\tauR,\tauV,\tauP$,
         parameters $k,W,S,c_{\min}$
\Ensure \textsc{Attack} or \textsc{Clean}
\State $\mathbf{q} \leftarrow f(Q)$;
       $\text{alarm\_count} \leftarrow 0$;
       $t \leftarrow 0$
\While{next $k$-word chunk $c$ available from $R$}
  \State $t \leftarrow t + 1$; \;
         $c_t \leftarrow c$; \;
         $\mathbf{e}_t \leftarrow f(c_t)$
  \State Compute $\RR_t$, $\VG_t$, $\TP_t$
         via \eqref{eq:rr}, \eqref{eq:vg}, \eqref{eq:tp}
  \If{$A_t = 1$}
    \State $\text{alarm\_count} \leftarrow \text{alarm\_count} + 1$
  \Else
    \State $\text{alarm\_count} \leftarrow 0$
  \EndIf
  \If{$\text{alarm\_count} = S$}
    \State \Return \textsc{Attack} \Comment{halt generation here}
  \EndIf
\EndWhile
\State \Return \textsc{Clean}
\end{algorithmic}
\end{algorithm}

When the deployment exposes streaming reasoning chunks and supports
client-side cancellation, \rg{} halts generation at the trigger
chunk, saving the remaining output tokens.
The token savings obtained from this early stop rule are reported
in \Cref{sec:deployment}.

\subsection{\heading{Query Drift Monitor}}
\label{sec:qdm}

\qdm{} is a post-hoc output-side companion monitor for deployment
settings where reasoning traces are hidden.
It splits the final output $Y$ into 80-word chunks
$\{y_1,\ldots,y_M\}$ and computes mean cosine similarity to the
query:
\begin{equation}
  \mathrm{QDM}(Y,Q) =
  \frac{1}{M}\sum_{m=1}^{M}\simop\!\left(f(y_m),\,f(Q)\right).
  \label{eq:qdm}
\end{equation}
A response is flagged if $\mathrm{QDM}(Y,Q) < 0.318$.
If $Y$ is empty or contains no non-empty chunk after stripping
whitespace, we set $\mathrm{QDM}(Y,Q){=}0$, so the response is
flagged by definition.
If $Y$ is non-empty but shorter than 80 words, the residual text is
treated as one chunk.
Unlike \rg{}, \qdm{} operates post-hoc on the final output and
does not use recurrence; it cannot enable early stopping.
Its threshold is derived from a development-set comparison of
clean and attacked outputs rather than from the mixed-benign
grid search used for \rg{}.
\qdm{} therefore uses an attack-informed development threshold and
should be interpreted as a fallback companion monitor rather than a
benign-only calibrated detector.
\qdm{} development uses the 50 clean SQuAD calibration examples
(indices 0--49) and matched OverThink development outputs from the
same questions; it does not use the held-out OverThink rows
(200--299), ExtendAttack rows (205--254), or adaptive C1/C4 rows
(255--354).
\qdm{} is not claimed as a primary novel contribution; it
provides fallback coverage on closed APIs where only final output
text is accessible.

\subsection{\heading{Calibration}}
\label{sec:calibration}

All \rg{} thresholds are calibrated on benign data exclusively.
We use a mixed-benign development set: 50 clean SQuAD traces
(indices 0--49) plus development subsets from HumanEval, GSM8K,
and CNN/DailyMail.
Percentile statistics over this set seed candidate threshold grids;
the final configuration is selected by grid search as the most
sensitive setting that achieves zero false positives across the
complete mixed-benign set.
Calibrating on mixed task types prevents the configuration from
overfitting to short QA outputs.
No attacked data appear in the calibration set; the thresholds
are locked before any test evaluation begins.

The \qdm{} threshold ($\tau{=}0.318$) is chosen from this
clean/attacked development comparison.
It is not iterated to optimise FPR and is treated as a companion
threshold, not as a primary calibrated detector.

\subsection{\heading{Runtime Behaviour}}
\label{sec:runtime}

\rg{} can operate as an online stop rule when two conditions hold:
the deployment exposes streaming reasoning chunks as they are
generated, and the client can issue a cancellation signal to halt
generation.
Under these conditions, \rg{} runs in the generation loop and
stops the model at the trigger chunk, converting a full-budget
generation into an early-terminated one.

\qdm{} cannot provide this capability because it requires the
complete output to be available before computing mean chunk
similarity.
For closed-API deployments with summarised reasoning, \rg{} can run
with reduced sensitivity and \qdm{} provides a post-hoc companion
check.
When reasoning text is hidden, \rg{} is structurally inapplicable
and \qdm{} is the fallback monitor.
The deployment trade-offs between \rg{} and \qdm{} across
different API visibility settings are discussed in
\Cref{sec:deployment}.

%% file: sections/05_theory.tex

\section{\heading{Theoretical Analysis}}
\label{sec:theory}

The results below are explanatory and falsifiable, not universal
security proofs.
They give geometric and probabilistic reasons for three observed
facts: TP dominance in the ablation, near-zero FPR under sustained
alarms, and the amplification collapse under full evasion.
Empirical Wilson confidence intervals (\Cref{sec:results}) provide
the operational performance certificates.

\subsection{\heading{Geometric View of Reasoning Collapse}}
\label{sec:geom}

Let $f : \text{text} \to \mathbb{R}^d$ be the sentence encoder
(\texttt{all-MiniLM-L6-v2}, $d{=}384$).
All embeddings are $\ell_2$-normalised, so cosine similarity equals
the inner product:
\[
  \simop(a,b) = \langle a,b \rangle, \qquad
  \|a\|_2 = \|b\|_2 = 1.
\]
Define the query embedding $\mathbf{q} = f(Q)$ and, for each
64-word reasoning chunk $c_t$, the chunk embedding
$\mathbf{e}_t = f(c_t)$.
Two scalar quantities track the trajectory:
\begin{align}
  a_t &= \simop(\mathbf{e}_t, \mathbf{q})
        \quad\text{(task alignment)},
        \label{eq:at}\\
  r_t &= \max_{i<t}\,\simop(\mathbf{e}_t,\mathbf{e}_i)
        \quad\text{(recurrence to past reasoning)}.
        \label{eq:rt}
\end{align}
Using the TP definition in \eqref{eq:tp}, we write
$\TP_t = a_t - r_t$.
Under clean reasoning, chunks introduce new task-relevant content,
keeping $a_t$ moderate and $r_t$ low.
Under an OverThink attack, after the decoy takes over, each chunk
resembles prior Sudoku-constraint chunks more than it resembles the
user query: $r_t$ is high, $a_t$ is low, so $\TP_t$ becomes
strongly negative.
We model this as an expected property of recurrence attacks; we do
not claim it holds for all possible attack constructions.

\subsection{\heading{Task-Conditioned Progress Separation}}
\label{sec:tp_sep}

\begin{theorem}[Task-conditioned progress separation]
\label{thm:tp_separation}
If $r_t - a_t \geq m$ for some $m > 0$, then $\TP_t \leq -m$.
Consequently, if $m > |\tauP|$, the TP alarm condition
$\TP_t < \tauP$ is satisfied at chunk $t$.
\end{theorem}

\begin{proof}
By definition~\eqref{eq:tp},
$\TP_t = a_t - r_t \leq -(r_t - a_t) \leq -m$.
\end{proof}

\noindent\textbf{Interpretation.}
The theorem states that TP fires when a chunk is both off-task
($a_t$ low) \emph{and} self-recurrent ($r_t$ high).
This conjunction distinguishes attack reasoning from legitimate
long-form generation.
In summarisation, chunks often introduce new source-document
content, so recurrence does not persist in the same way as a
decoy loop.
In our CNN/DailyMail evaluation, TP does not remain below $\tauP$
for $S{=}3$ consecutive chunks, while \qdm{} fires because it only
observes low query similarity.
The ablation shows the same pattern from another angle: TP-only
achieves 99\% TPR and 0\% FPR on the primary OverThink evaluation,
while RR-only and VG-only show higher FPR.

\subsection{\heading{Sustained Alarm and False-Positive Control}}
\label{sec:sustained}

Define the per-chunk alarm indicator:
\begin{equation}
  A_t = \ind\!\left[
    \RR_t > \tauR \;\land\;
    \VG_t < \tauV \;\land\;
    \TP_t < \tauP \;\land\;
    t \geq c_{\min}
  \right].
\end{equation}
The detector fires when $A_t = 1$ for $S$ consecutive chunks.

\begin{proposition}[Sustained alarm false-positive decay]
\label{prop:fpr_decay}
Let $p_0 \in (0,1)$ be an upper bound on the probability that all
three alarm conditions hold simultaneously for a single chunk under
benign reasoning.
Under weak temporal dependence with correction factor $\eta \geq 1$,
the probability of $S$ consecutive alarms satisfies:
\[
  \Pr\!\left(\bigcap_{j=0}^{S-1} \{A_{t+j}=1\}\right)
  \;\leq\; \eta\cdot p_0^{\,S}.
\]
\end{proposition}

\begin{proof}
By assumption, each joint alarm event has marginal probability at
most $p_0$.
Under weak dependence, the probability of $S$ consecutive events
is bounded by $\eta p_0^S$ where $\eta \geq 1$ absorbs the
dependence structure.
\end{proof}

\begin{remark}
The threshold $\tauP{=}{-0.738}$ requires $\TP_t$ to be strongly
negative; $\tauR{=}0.20$ requires elevated recurrence; and
$\tauV{=}0.01$ requires near-zero embedding-cloud growth.
This joint event is geometrically restrictive for normal reasoning
text: all three alarm conditions must hold together for $S{=}3$
consecutive chunks.
Requiring $S{=}3$ multiplies the bound by $p_0^2$ relative to
$S{=}1$, giving exponential suppression.

This proposition is a \emph{stylised explanatory bound}.
Real reasoning text is not isotropic; $p_0$ is not derived from
first principles but is bounded empirically.
The bound explains \emph{why $S$ matters}, not what the exact
false-positive rate is.
Thresholds are set by mixed-benign grid search
(\Cref{sec:calibration_proc}), and the actual FPR is certified by
the held-out experiment in \Cref{prop:certificate}.
\end{remark}

\subsection{\heading{Evasion-Alignment Tradeoff}}
\label{sec:evasion_theory}

\begin{theorem}[Evasion-alignment condition]
\label{thm:evasion}
For any chunk $t$ at which the TP alarm does not fire
($\TP_t \geq \tauP$), the attacker must satisfy:
\begin{equation}
  a_t \;\geq\; r_t + \tauP.
  \label{eq:evasion}
\end{equation}
\end{theorem}

\begin{proof}
From~\eqref{eq:tp}, $\TP_t = a_t - r_t \geq \tauP$ implies
$a_t \geq r_t + \tauP$.
\end{proof}

\noindent\textbf{Interpretation.}
Since $\tauP = -0.738 < 0$, condition~\eqref{eq:evasion}
requires task alignment $a_t$ to remain sufficiently close to the
recurrence level $r_t$.
With respect to TP, a fully evading attacker cannot loop freely in
an unrelated semantic region: it must either reduce recurrence or
increase task alignment.
The C4: Task-aligned filler attack takes the second path, staying close
to the user query, which causes the model to converge on an
answer rather than loop indefinitely.
Amplification collapses from $22.8\times$ (OverThink) to
$2.2\times$ (C4), consistent with this tradeoff.

\noindent\textbf{Scope.}
\Cref{thm:evasion} constrains \emph{full} TP evasion, where
$\TP_t \geq \tauP$ at every chunk.
It does not rule out partial evasion, which is exactly the C1
boundary observed in \Cref{sec:adaptive}: C1 evades the TP alarm
on enough chunks to achieve $\approx$50\% joint miss at
$11.9\times$ amplification, a setting the theorem does not bound.

\subsection{\heading{Empirical FPR Certificate}}
\label{sec:certificate}

\begin{proposition}[Held-out false-positive certificate]
\label{prop:certificate}
\rg{} produces $k{=}0$ false positives on $n{=}150$ fully
held-out benign SQuAD queries (indices 50--199, shuffled with
\texttt{seed=42}).
The Wilson 95\% upper confidence bound on the true FPR is
$\hat{p}_\text{upper} < 2.5\%$.
\end{proposition}

This certificate is empirical, not derived from the isotropic null
model of Proposition~\ref{prop:fpr_decay}.
In the broader clean-task profile, \rg{} also achieves ${\leq}1\%$
FPR on HumanEval ($N{=}100$), 0\% on GSM8K ($N{=}100$), and
0\% on CNN/DailyMail ($N{=}100$).
Those rows test clean-workload behaviour rather than extending the
independent SQuAD certificate, because the long-form tasks also
contribute development examples to mixed-benign calibration.
The pattern is still consistent with the prediction that TP does
not fire when $r_t$ is low.

%% file: sections/06_setup.tex

\section{\heading{Experimental Setup}}
\label{sec:setup}

\subsection{\heading{Models}}
\label{sec:models}

We evaluated four open-weight models and two closed-API models.
Open-weight experiments run on NVIDIA A100 GPUs using
HuggingFace Transformers~\cite{wolf2020transformers} and
SentenceTransformers~\cite{reimers2019sbert}.
Closed-API experiments access the Anthropic Messages API with
extended thinking enabled.
Anthropic documents that extended thinking can expose varying levels
of transparency, including summarised thinking blocks~\cite{anthropic2026thinking};
the evaluated Claude models are documented in Anthropic's system-card
index~\cite{anthropic2026systemcards}.
Table~\ref{tab:models} summarises all the models.

\begin{table}[H]
\centering
\caption{Models evaluated.
Trace visibility: whether the reasoning trace is accessible to the defender.}
\label{tab:models}
\footnotesize
\setlength{\tabcolsep}{2pt}
\renewcommand{\arraystretch}{1.05}
\begin{tabularx}{\columnwidth}{@{}>{\raggedright\arraybackslash}Xlcc@{}}
\toprule
Model & Type & Reasoning & Trace visibility \\
\midrule
DS-R1-Qwen-7B~\cite{deepseek2025r1}    & Open primary & Yes & Full \\
DS-R1-Llama-8B~\cite{deepseek2025r1}   & Open         & Yes & Full \\
Qwen3-8B~\cite{qwen2025tech}           & Open         & Yes & Full \\
Llama-3.1-8B~\cite{meta2024llama3}     & Open non-reas. & No  & N/A  \\
\midrule
Sonnet~4.5\textsuperscript{a}          & Closed API   & Yes & Summary \\
Opus~4.7\textsuperscript{b}            & Closed API   & Yes  & Hidden \\
\bottomrule
\multicolumn{4}{@{}p{\columnwidth}@{}}{\footnotesize
  Full = exposed reasoning trace.
  \textsuperscript{a}API ID: \texttt{claude-sonnet-4-5-20250929}.
  Thinking blocks are provider-summarised, not raw reasoning.}\\
\multicolumn{4}{@{}p{\columnwidth}@{}}{\footnotesize
  \textsuperscript{b}Thinking blocks were empty in all our runs;
  \rg{} is inapplicable for Opus.}
\end{tabularx}
\end{table}

\noindent\textbf{Primary model.}
DeepSeek-R1-Distill-Qwen-7B (\emph{ds-r1-qwen-7b}) is used for the
complete evaluation: main-attack characterisation, \rg{} ablation,
adaptive experiments, and baseline comparison.
It is chosen because it generates long, accessible reasoning traces
and shows clear liveness failure under both primary attacks.

\noindent\textbf{Cross-model validation.}
The three additional open-weight models validate that attack patterns
and detection logic generalise across model families.
Llama-3.1-8B-Instruct is a non-reasoning model included to establish
the boundary of \rg{}'s applicability: without extended reasoning
text the three signals are undefined.

\noindent\textbf{Closed-API evaluation.}
Sonnet~4.5 is evaluated with $N{=}30$ attacked and $N{=}30$ clean
queries.
Its visibility is Regime~B, and the attacked outputs exhibit O2
slowdown: all queries produce answers at $21.3\times$ token
amplification.
Opus~4.7 is a \emph{qualitative smoke test only} ($N{=}3$).
No statistical claims are drawn from the Opus results
(\Cref{sec:closed_api_protocol,sec:deployment}).

\subsection{\heading{Datasets and Splits}}
\label{sec:splits}

\textbf{SQuAD}~\cite{rajpurkar2016squad} is the main evaluation
dataset.
We apply a fixed shuffle (\texttt{seed=42}) to the SQuAD validation
split and allocate examples to the role-specific ranges in
Table~\ref{tab:splits}.
The adaptive extension deliberately overlaps indices 255--299 with
the held-out attack range, as described below.
Table~\ref{tab:splits} lists the complete allocation.

\begin{table}[t]
\centering
\caption{SQuAD validation split allocation (shuffled, seed\,=\,42).
Adaptive extension partially overlaps with the attack test range;
see text for details.}
\label{tab:splits}
\footnotesize
\setlength{\tabcolsep}{3pt}
\begin{tabularx}{\columnwidth}{@{}llrX@{}}
\toprule
Split & Indices & $N$ & Purpose \\
\midrule
Calibration          & 0--49   &  50 & Threshold selection only \\
Held-out benign      & 50--199 & 150 & Main FPR evaluation \\
Held-out attack      & 200--299 & 100 & TPR tests \\
Adaptive extension   & 255--354 & 100 & C1/C4 tests \\
\bottomrule
\end{tabularx}
\end{table}

Three details matter for interpreting this split.

\textit{Calibration split and FPR certificate.}
Indices 0--49 are used only to select detector thresholds.
They are not attack test data, and they are not counted as
independent held-out FPR evidence.
The primary held-out SQuAD FPR certificate is measured on indices
50--199 ($N{=}150$), where \rg{} produces 0 false positives
(\Cref{prop:certificate}).

\textit{Attack and adaptive overlap.}
Indices 255--299 appear in both the held-out attack test (200--299)
and the adaptive extension (255--354).
The same SQuAD questions appear in both sets, but the attack prompt
is different.
No OverThink result was used to design C1 or C4.

\textit{ExtendAttack indices.}
ExtendAttack uses indices 205--254 ($N{=}50$ per model),
drawn from the held-out attack test range.

\medskip
\noindent\textbf{Long-form benign benchmarks.}
For \textbf{HumanEval}~\cite{chen2021humaneval},
\textbf{GSM8K}~\cite{cobbe2021gsm8k}, and
\textbf{CNN/DailyMail}~\cite{see2017cnn,hermann2015cnn}, each clean-task profile
contains 100 examples.
For \rg{} calibration, 30 examples per task are used as
mixed-benign development data and the remaining 70 examples per
task are held out in the artifact.
The main table reports the full 100-example clean-task profile so
all defences are compared on the same workload.
These benchmarks contain clean examples only; the independent
statistical FPR certificate is the SQuAD $0/150$ result in
\Cref{prop:certificate}.

\noindent\textbf{TriviaQA.}
TriviaQA~\cite{joshi2017triviaqa} ($N{=}25$ clean, $N{=}25$ attacked)
provides cross-dataset validation reported in the supplementary artifact.

\subsection{\heading{Attack Families}}
\label{sec:attack_setup}

We evaluate \rg{} and all baselines against two previously
proposed reasoning-chain consumption attacks and two scaled
adaptive variants.

\textbf{OverThink}~\cite{kumar2025overthink} injects a Sudoku-puzzle
decoy with an explicit instruction to solve it before answering.
We reproduce this attack under our own models and evaluation setting
with $N{=}100$ attacked queries per open model (indices 200--299).

\textbf{ExtendAttack}~\cite{extendattack2025} encodes the user's
question in mixed-base ASCII and instructs the model to decode it
character-by-character.
We evaluate it with $N{=}50$ per open model (indices 205--254).

\textbf{Adaptive variants C1: Topical CSP and C4: Task-aligned filler}~(\Cref{sec:adaptive}) are
evaluated with $N{=}100$ each (indices 255--354) on
DS-R1-Qwen-7B only.
Two further variants, C2: High-entropy filler and
C3: Structured-diverse filler, were smoke-tested at $N{=}3$;
both produced $\leq 1.5\times$ amplification and are not scaled.

\subsection{\heading{Defenses}}
\label{sec:defenses_setup}

We evaluate six main defences in the primary benchmark:
\rg{}, \qdm{}, and B1--B4.
We additionally evaluate Llama Guard~3-1B-INT4 as an input-side
safety classification baseline.
Within the primary benchmark, \rg{} is our main proposed monitor;
\qdm{} is a lightweight output-side companion monitor we implement
to cover closed-API settings where reasoning traces are inaccessible;
B1--B4 are reference baselines.

\textbf{B1: Token Budget.}
Flags outputs whose token count exceeds the 99th percentile of
clean calibration outputs.

\textbf{B2: Compression-Ratio Proxy.}
Uses zlib compression ratio as a proxy for output repetitiveness:
highly repetitive outputs compress more, indicating low semantic
diversity.
We use this proxy rather than true model perplexity because
closed APIs and saved text outputs do not expose token-level logits.

\textbf{B3: Answer-Absence Proxy.}
Flags outputs with fewer than 5 answer tokens as liveness failures,
loosely inspired by answer-convergence monitoring~\cite{liu2025convergence}.

\textbf{B4: Output-Length Anomaly.}
Flags outputs whose token count has a z-score $>3$ relative to the
calibration mean.

\textbf{\qdm{} (Query Drift Monitor).}
Computes the mean cosine similarity between 80-word output chunks
and the user query using all-MiniLM-L6-v2; flags output if
similarity falls below $\tau_{\qdm}{=}0.318$ (full description in
\Cref{sec:qdm}).
\qdm{} is a lightweight companion monitor we implement for
deployment scenarios where the reasoning trace is inaccessible.
It is not claimed as a primary novel contribution.
Unlike embedding-drift detectors that compare two input
representations~\cite{zedd2025}, \qdm{} compares the user query
against the model's \emph{output}, making it a post-hoc
liveness monitor rather than an input-side filter.

\textbf{\rg{}.}
Our main proposed method (\Cref{sec:method}).
Processes the exposed reasoning trace in 64-word chunks, computes
three signals (RR, VG, TP), and raises a sustained alarm when all
three calibrated alarm conditions hold for $S{=}3$ consecutive chunks.
Requires non-empty reasoning text; inapplicable when the API hides
or returns empty thinking blocks.
Unlike all other defences listed here, \rg{} enables \emph{online
early stopping}: generation can be halted at the trigger chunk,
saving about 92\% of tokens per detected attack.

\textbf{Llama Guard~3-1B-INT4}~\cite{llamaguard3} ($N{=}50$ attack,
$N{=}50$ clean) is evaluated as an input-side safety classification
baseline (\Cref{sec:deployment}) to establish whether standard
content moderation can intercept this attack class at the input stage.

Table~\ref{tab:defense_comparison} summarises the capability
differences between the two main monitors.

\begin{table}[t]
\centering
\caption{Visibility-specific capability comparison: \rg{} vs.\ \qdm{}.}
\label{tab:defense_comparison}
\footnotesize
\setlength{\tabcolsep}{3pt}
\renewcommand{\arraystretch}{1.05}
\begin{tabularx}{\columnwidth}{@{}lcccXX@{}}
\toprule
Monitor & Needs reas. & Online & Early stop & Final-output only & Summ. thinking \\
\midrule
\rg{}  & Yes & Yes & Yes & No  & Partial \\
\qdm{} & No  & No  & No  & Yes & Yes \\
\bottomrule
\end{tabularx}
\end{table}

\subsection{\heading{RecurGuard Configuration}}
\label{sec:rg_config}

\rg{} uses one fixed configuration throughout the evaluation.
Chunk size $k{=}64$~words; window $W{=}20$~chunks;
inner similarity threshold $\taur{=}0.70$.
Calibrated alarm thresholds:
$\tauR{=}0.20$, $\tauV{=}0.01$, $\tauP{=}{-0.738}$;
sustained alarm count $S{=}3$; minimum chunk count $c_{\min}{=}3$.
Sentence encoder: \texttt{all-MiniLM-L6-v2} (384 dimensions).
The same locked configuration is used for every test evaluation.

\subsection{\heading{Calibration Procedure}}
\label{sec:calibration_proc}

\textbf{\rg{} thresholds} are calibrated on benign data only.
No attacked samples appear in the \rg{} calibration set.

We use a \emph{mixed-benign calibration set}: 50 SQuAD traces
(indices 0--49) plus 30 benign long-form traces from each of
HumanEval, GSM8K, and CNN/DailyMail.
Percentile statistics over this set seed candidate threshold grids;
the final configuration is selected by grid search as the most
sensitive setting achieving zero false positives on the complete
mixed-benign set.
Calibrating on mixed task types prevents the configuration from
overfitting to short QA outputs and supports near-zero FPR on
legitimate long-form reasoning.

\textbf{\qdm{} threshold.}
$\tau_{\qdm}{=}0.318$ is chosen from a development comparison of
clean and attacked outputs.
\qdm{} is treated as a companion monitor, not the primary calibrated
detector.
Because this threshold uses attacked development outputs, \qdm{} is
not a benign-only calibrated detector like \rg{}.
The attacked development outputs are OverThink generations on the
SQuAD calibration questions (indices 0--49), disjoint from the
OverThink, ExtendAttack, and C1/C4 test rows.
All thresholds are locked after calibration and never retuned.

\subsection{\heading{Metrics}}
\label{sec:metrics}

\textbf{Token amplification}: ratio of mean attacked output tokens
to mean clean output tokens.

\textbf{Liveness failure rate}: fraction of attacked queries in
which the model exhausts the generation budget without producing
a final answer (fewer than 5 answer tokens, O3).

\textbf{True positive rate (TPR)}: fraction of attacked queries
correctly flagged.

\textbf{False positive rate (FPR)}: fraction of clean queries
incorrectly flagged.

\textbf{Joint miss rate}: fraction of attacked queries missed by
\emph{both} \rg{} and \qdm{} simultaneously.

Wilson 95\% confidence intervals are reported for primary statistical
claims and small-sample detection results~\cite{wilson1927probable}.
Paired comparison between \qdm{} and \rg{} on the same samples
uses the \textbf{McNemar test}~\cite{mcnemar1947sampling}.

\subsection{\heading{Closed-API Evaluation Protocol}}
\label{sec:closed_api_protocol}

\textbf{Sonnet~4.5.}
Evaluated with $N{=}30$ clean and $N{=}30$ attacked queries.
We report Wilson confidence intervals for this sample size, but
treat the result as a limited closed-API evaluation rather than a
full model family claim.
Thinking blocks are accessible but contain summarised reasoning;
\rg{} is applicable with the reduced-sensitivity caveat noted in
\Cref{sec:visibility,sec:models}.

\textbf{Opus~4.7.}
Evaluated with $N{=}3$ attacked queries as a qualitative smoke test.
We do not report confidence intervals or draw statistical claims
from Opus.
The API returned empty thinking blocks for all samples;
\rg{} is structurally inapplicable regardless of $N$.
Additional samples would test the same visibility limitation rather
than the detector itself.

\noindent\textbf{Artifact and hardware.}
Open-weight experiments were run on NVIDIA A100 GPUs.
The complete artifact is available at
\url{https://github.com/abidaziz1/recurguard}.

%% file: sections/07_results.tex

\section{\heading{Evaluation Results}}
\label{sec:results}

\subsection{\heading{Attack Amplification}}
\label{sec:results_amplification}

Table~\ref{tab:amplification} reports token amplification and
liveness failure rates across all four open models for both primary
attack families.
OverThink drives all three reasoning-specialist models into
O3 liveness failure on 86--100\% of queries, with 22--24$\times$
token amplification.
ExtendAttack produces comparable amplification (19--23$\times$) via
a structurally different mechanism: mixed-base ASCII decoding rather
than puzzle injection.
The comparable amplification of two mechanistically distinct attacks
suggests that liveness failure is not tied to a single decoy
construction.

Llama-3.1-8B-Instruct shows a nominal amplification of
$402\times$ under OverThink, but this is a denominator artefact:
its clean baseline is only 14 tokens (the model answers without
generating extended reasoning), making any non-trivial attacked
output appear as extreme amplification.
We do not interpret this as a stronger attack on Llama; it reflects
the absence of a reasoning trace rather than a meaningful
consumption threat.

\input{tables/tab_attack_severity}

\subsection{\heading{Primary Attack Detection}}
\label{sec:results_detection}

Table~\ref{tab:benchmark} (upper) shows true positive rates for all
six defences on the two primary attacks.
\input{tables/tab_main_benchmark}
On both OverThink and ExtendAttack, B1 and B4 achieve 100\% TPR
because both attacks produce outputs that far exceed the
token-budget and length-anomaly thresholds calibrated on clean data.
For these primary attacks, length alone is a reliable signal.

\rg{} achieves 99\% TPR on OverThink and 92\% on ExtendAttack.
\qdm{} achieves 100\% and 96\% respectively.
A McNemar test confirms that \rg{} and \qdm{} agree on 99/100
OverThink samples ($\chi^2{=}0.0$, $p{=}1.0$): the two methods
are not statistically distinguishable in TPR on primary attacks.
The sharper difference appears in false positives and deployment
behaviour, described next.

\subsection{\heading{False Positive Rate Across Task Types}}
\label{sec:results_fpr}

Clean SQuAD is not the hard case.
All six defences stay at ${\leq}1\%$ FPR on question answering.
The harder test is whether a detector can leave normal long-form
answers alone.

HumanEval exposes the problem with length rules.
B1 flags 69\% of clean code answers and B4 flags 80\%.
Those are not malicious outputs; they are ordinary code-generation
answers that happen to be long and structured.
A threshold calibrated on short QA answers is too blunt for this
workload.

CNN/DailyMail exposes a different weakness.
\qdm{} flags 89\% of clean summaries because its anchor is the
short instruction, ``Summarise this article.''
The summary mostly discusses the article, so its chunks can sit far
from the instruction vector even when the answer is correct.

\rg{} avoids both failure modes in this evaluation.
It produces 0\% FPR on the held-out SQuAD split ($N{=}150$),
1\% on HumanEval ($N{=}100$), 0\% on GSM8K ($N{=}100$), and
0\% on CNN/DailyMail ($N{=}100$).
The SQuAD result is the independent certificate: 0 false positives
on indices 50--199, with Wilson 95\% upper bound $<2.5\%$.
The long-form rows are clean-workload checks, not a second
independent certificate, because those tasks also provide development
examples for mixed-benign calibration.

The summarisation result checks the TP signal directly.
A clean summary may be far from the query, but it keeps introducing
new article content, so recurrence to prior chunks stays low.
That is the distinction behind \Cref{thm:tp_separation}: \rg{}
should not fire merely because the answer is semantically far from
a short instruction.

\Cref{fig:combined_results} gives a compact view of the detector behaviour
and the adaptive boundary. The left panel illustrates how RR, VG, and TP move
into the sustained joint-alarm region; the right panel shows that full
semantic evasion is possible only after attack amplification drops sharply.

\begin{figure*}[!t]
\centering
\includegraphics[width=\textwidth]{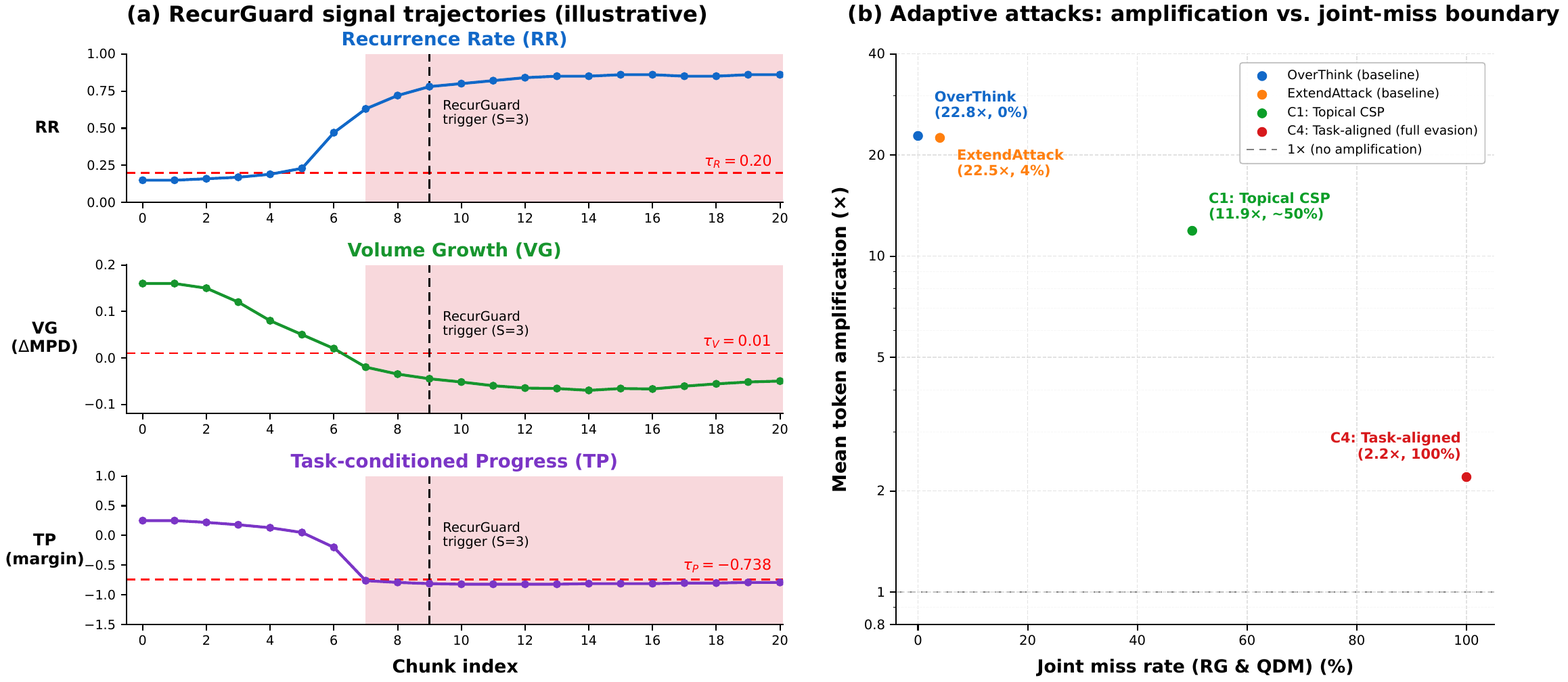}
\caption{Combined detector view. Left: illustrative \rg{} signal
trajectories showing recurrence growth, reduced volume growth, and
negative task-conditioned progress before the sustained $S{=}3$ alarm.
Right: adaptive amplification versus joint miss rate on DS-R1-Qwen-7B.
C1 remains the middle case, while C4 fully evades the semantic monitors
only after amplification drops to $2.2\times$.}
\label{fig:combined_results}
\end{figure*}

\subsection{\heading{Baseline Defence Comparison}}
\label{sec:results_baselines}

The adaptive rows are where the baselines separate.

\textbf{Length checks still catch C1.}
B1 detects 94\% of C1 attacks and B4 detects 99\%.
C1 stays expensive even after the decoy is made topical, so the
output is still long.
The same rules barely catch C4 (B1: 3\%, B4: 10\%) because C4
brings amplification down to $2.2\times$.
The weak point is deployment: B1 and B4 also flag 69\% and 80\% of
clean HumanEval answers.
That false-positive rate is too high for a mixed workload.

\textbf{B2 fails once the output stops looking repetitive.}
The compression-ratio proxy is strong on OverThink and ExtendAttack,
but drops to 3\% TPR on C1 and 1\% on C4.
The saved outputs explain why.
C1 has mean compression ratio 0.206 and C4 has mean ratio 0.386,
both above the flag threshold.
The adaptive prompts diversify the text enough to remove the
compression signature.

\textbf{B3 catches only missing answers.}
The answer-absence proxy reaches 31\% TPR on C1 because some C1
runs exhaust the budget before producing a final answer.
It reaches only 1\% on C4.
That is expected: C4 usually produces a short answer, so there is
no answer-absence signal to detect.

\Cref{sec:adaptive} uses these failures to draw the defence
taxonomy and security boundary.

\subsection{\heading{Signal Ablation}}
\label{sec:ablation}

\input{tables/tab_ablation}

\Cref{tab:ablation} reports the seven signal configurations tested on
OverThink.
The ablation is simple: TP carries most of the detector's precision.
On OverThink ($N{=}100$, DS-R1-Qwen-7B), TP alone matches the full
model at 99\% TPR and 0\% FPR.
RR alone and VG alone catch all attacks, but they also false-alarm
on clean data: 4.7\% FPR for RR and 5.3\% for VG.
Every configuration with TP reaches 0\% FPR.
Every configuration without TP has elevated FPR.

The reason is the same one used in the theory section.
RR and VG can tell when reasoning is becoming repetitive or no
longer expanding, but they do not know whether the repetition is
off-task.
TP adds that missing comparison to the query.
That is why it can reject benign long-form reasoning while still
catching decoy collapse.

The FPR certificate remains the SQuAD held-out test:
0/150 false positives on indices 50--199, with Wilson 95\% upper
bound $<2.5\%$.
The HumanEval, GSM8K, and CNN/DailyMail rows ($N{=}100$ each) show
the broader clean-task profile; \rg{} stays at ${\leq}1\%$ FPR on
all three.

\subsection{\heading{Closed-API Evaluation}}
\label{sec:results_closed}

\input{tables/tab_closed_api}
\input{tables/tab_adaptive}

The closed-API results need a separate reading because the defender
does not see the same evidence as in the open-weight model runs.

\textbf{Sonnet~4.5} ($N{=}30$ attacked, $N{=}30$ clean) has
visibility Regime~B and outcome O2: all 30 attacked queries
produce a final answer, but at $21.3\times$ token amplification.
\qdm{} detects 100\% of attacked queries [95\% CI: 88.6--100\%],
with 0\% FPR on clean queries.
\rg{} detects 16.7\% [7.3--33.6\%], attributable to Sonnet's
provider-summarised thinking blocks rather than raw reasoning text:
summarisation compresses the trace and dampens chunk-level
recurrence signals.
For closed-API deployments with partial-visibility thinking blocks,
\qdm{} is the more reliable monitor.

\textbf{Opus~4.7} ($N{=}3$, qualitative only) has visibility
Regime~C and outcome O3: all three attacked queries exhaust the
generation budget and produce no answer, with $676\times$ token
amplification.
The cost implication of this amplification is analysed in
\Cref{sec:deployment}.
\qdm{} flags all three samples (3/3).
Because $N{=}3$ is a qualitative smoke test, we do not report a
confidence interval or draw statistical claims from Opus.
\rg{} is inapplicable: the API returned empty thinking blocks for
all attacked queries.
These observations establish only that the API visibility boundary
makes \rg{}
structurally inapplicable when thinking blocks are hidden.

The closed-API results confirm the deployment boundary stated in
\Cref{sec:method}: when reasoning is summarised, \rg{} can run with
reduced sensitivity; when reasoning is hidden, \rg{} is structurally
inapplicable.
\qdm{} remains the post-hoc fallback in both cases.

%% file: tables/tab_attack_severity.tex

\begin{table*}[!t]
\centering
\caption{Attack amplification across models.
         OT\,=\,OverThink; EA\,=\,ExtendAttack; LF\,=\,liveness failure (O3).}
\label{tab:amplification}
\setlength{\tabcolsep}{5pt}
\begin{tabular}{@{}lrrrccc@{}}
\toprule
\multirow{2}{*}{Model}
  & \multicolumn{2}{c}{Mean tokens}
  & \multicolumn{2}{c}{OverThink}
  & \multicolumn{2}{c}{ExtendAttack} \\
\cmidrule(lr){2-3}\cmidrule(lr){4-5}\cmidrule(lr){6-7}
  & Clean & OT & Amp. & LF & Amp. & LF \\
\midrule
DS-R1-Qwen-7B   & 357 & 8{,}155 & 22.8$\times$ & 98\% & 22.5$\times$ & 92\% \\
DS-R1-Llama-8B  & 344 & 8{,}016 & 23.3$\times$ & 86\% & 22.7$\times$ & 80\% \\
Qwen3-8B        & 337 & 8{,}192 & 24.3$\times$ & 100\%& 19.6$\times$ & 44\% \\
Llama-3.1-8B$^\dagger$ & 14  & 5{,}701 & 402$\times$ & 100\%& 212$\times$ & 100\%\\
\bottomrule
\multicolumn{7}{@{}l@{}}{\footnotesize
  $^\dagger$Non-reasoning model; 402$\times$ is a denominator artefact
  (14-token clean baseline).}\\
\multicolumn{7}{@{}l@{}}{\footnotesize
  Mean token counts are rounded; amplification is computed from
  unrounded means.}
\end{tabular}
\end{table*}

%% file: tables/tab_main_benchmark.tex

\begin{table*}[!t]
\centering
\caption{Defence benchmark on DS-R1-Qwen-7B.
         Upper block: TPR; lower block: FPR.
         \textcolor{red}{Red} marks FPR $>$10\%.}
\label{tab:benchmark}
\setlength{\tabcolsep}{7pt}
\begin{tabular}{@{}lrrrrrr@{}}
\toprule
& \multicolumn{6}{c}{Defence} \\
\cmidrule(l){2-7}
Attack / Clean task & B1 & B2 & B3 & B4 & QDM & \rg{} \\
\midrule
\multicolumn{7}{@{}l@{}}
  {\small\textit{True positive rate (TPR): higher is better}} \\[2pt]
OverThink ($N$\,=\,100) &
  \textbf{100\%} & \textbf{100\%} & 98\% & \textbf{100\%} &
  \textbf{100\%} & 99\% \\
ExtendAttack ($N$\,=\,50) &
  \textbf{100\%} & 96\% & 92\% & \textbf{100\%} & 96\% & 92\% \\
C1: Topical CSP ($N$\,=\,100) &
  94\% & 3\%\,$^\star$ & 31\% & \textbf{99\%} & 38\% & 21\% \\
C4: Task-aligned filler ($N$\,=\,100)$^\ddagger$ &
  3\% & 1\%\,$^\star$ & 1\% & \textbf{10\%} & 0\% & 0\% \\
\midrule
\multicolumn{7}{@{}l@{}}
  {\small\textit{False positive rate (FPR): lower is better}} \\[2pt]
SQuAD QA full profile ($N$\,=\,200; incl.\ 50 calib.) &
  1\% & 1\% & 0.5\% & 1\% & 0.5\% & \textbf{0\%} \\
HumanEval code ($N$\,=\,100) &
  \textcolor{red}{69\%} & \textcolor{red}{27\%} &
  \textcolor{red}{17\%} & \textcolor{red}{80\%} &
  \textbf{1\%} & \textbf{1\%} \\
GSM8K math ($N$\,=\,100) &
  \textbf{0\%} & \textbf{0\%} & \textbf{0\%} & \textbf{0\%} &
  \textbf{0\%} & \textbf{0\%} \\
CNN/DM summary ($N$\,=\,100) &
  \textbf{0\%} & 2\% & \textbf{0\%} & 1\% &
  \textcolor{red}{89\%} & \textbf{0\%} \\
\bottomrule
\multicolumn{7}{@{}l@{}}{\footnotesize
  B1\,=\,Token Budget; B2\,=\,Compression-Ratio Proxy;
  B3\,=\,Answer-Absence Proxy; B4\,=\,Output-Length Anomaly.}\\
\multicolumn{7}{@{}l@{}}{\footnotesize
  $^\star$B2 threshold is a fallback value; adaptive B2 rows should be
  read with caution. $^\ddagger$C4 full evasion at 2.2$\times$.}\\
\multicolumn{7}{@{}l@{}}{\footnotesize
  \rg{} is the only defence with near-zero FPR across all four clean
  task types.}\\
\multicolumn{7}{@{}l@{}}{\footnotesize
  \rg{} held-out certificate: 0/150 false positives on SQuAD indices
  50--199 (\Cref{prop:certificate}).}\\
\multicolumn{7}{@{}l@{}}{\footnotesize
  Long-form FPR rows report the full 100-example clean-task profile;
  30 examples per long-form task also support mixed-benign calibration.}
\end{tabular}
\end{table*}

%% file: tables/tab_ablation.tex

\begin{table}[!t]
\centering
\caption{Signal ablation on OverThink ($N$\,=\,100, DS-R1-Qwen-7B).
         \checkmark\,=\,signal active; --\,=\,disabled.
         TP is the dominant signal; any configuration including TP
         achieves 0\% FPR.}
\label{tab:ablation}
\setlength{\tabcolsep}{6pt}
\begin{tabular}{@{}lccccc@{}}
\toprule
Configuration & RR & VG & TP & TPR & FPR \\
\midrule
RR-only       & \checkmark & -- & -- & 100\% & 4.7\% \\
VG-only       & -- & \checkmark & -- & 100\% & 5.3\% \\
\textbf{TP-only} & -- & -- & \checkmark & \textbf{99\%} & \textbf{0\%} \\
\midrule
RR\,+\,VG     & \checkmark & \checkmark & -- & 100\% & 4.0\% \\
RR\,+\,TP     & \checkmark & -- & \checkmark & 99\% & \textbf{0\%} \\
VG\,+\,TP     & -- & \checkmark & \checkmark & 99\% & \textbf{0\%} \\
\textbf{Full (all)} & \checkmark & \checkmark & \checkmark
             & \textbf{99\%} & \textbf{0\%} \\
\bottomrule
\end{tabular}
\end{table}

%% file: tables/tab_closed_api.tex

\begin{table}[!t]
\centering
\caption{Closed-API evaluation.
         Sonnet~4.5: $N{=}30$ attacked, $N{=}30$ clean.
         Opus~4.7: $N{=}3$ qualitative smoke test only; no
         statistical claims.
         Vis.\,=\,visibility regime; Out.\,=\,attack outcome
         (O2 slowdown, O3 liveness failure).
         CI\,=\,95\% Wilson interval.
         RG\,=\,\rg{}; ``N/A''\,=\,not applicable (empty thinking blocks).}
\label{tab:closed}
\scriptsize
\setlength{\tabcolsep}{2pt}
\renewcommand{\arraystretch}{1.08}
\begin{tabular}{@{}lccrcc@{}}
\toprule
Model & Vis. & Out. & Amp. & QDM TPR & RG TPR \\
\midrule
Sonnet~4.5 &
  B &
  O2 &
  21.3$\times$ &
  \textbf{100\%}\;[88.6, 100] &
  16.7\%\;[7.3, 33.6] \\
Opus~4.7 &
  C &
  O3 &
  676$\times$ &
  100\% (3/3) &
  N/A \\
\bottomrule
\end{tabular}
\end{table}

%% file: tables/tab_adaptive.tex

\begin{table*}[!t]
\centering
\caption{Adaptive adversary results (DS-R1-Qwen-7B unless noted).
         Wilson confidence intervals are shown in brackets where reported.
         Joint miss = neither monitor detected the attack,
         Amp.\ = mean token amplification.}
\label{tab:adaptive}
\setlength{\tabcolsep}{4.5pt}
\renewcommand{\arraystretch}{1.08}
\begin{tabular}{@{}lrrrrr@{}}
\toprule
Attack & $N$ & Amp. & RG TPR & QDM TPR & Joint miss \\
\midrule
OverThink (baseline)    & 100 & 22.8$\times$ &
  99\%\;[94.6, 99.8] & 100\% & 0\% \\
ExtendAttack (baseline) &  50 & 22.5$\times$ &
  92\%\;[81.2, 96.8] & 96\% & 4\% \\
\midrule
C1: Topical CSP         & 100 & 11.9$\times$ &
  21\%\;[14.2, 30.0] & 38\%\;[29.1, 47.8] & 50\%\,$^\dagger$ \\
C4: Task-aligned filler & 100 & \textbf{2.2$\times$} &
  0\%\;[0, 3.7]      & 0\%\;[0, 3.7]      & 100\% \\
\midrule
C2: High-entropy filler$^\S$ & 3 & $\leq$1.5$\times$ & n/a & n/a & n/a \\
C3: Structured-diverse filler$^\S$ & 3 & $\leq$1.5$\times$ & n/a & n/a & n/a \\
\bottomrule
\multicolumn{6}{@{}l@{}}{\footnotesize
  $^\S$Smoke test only; model processed both decoys and answered correctly.}\\
\multicolumn{6}{@{}l@{}}{\footnotesize
  $^\dagger$Rounded locked adaptive run-level result.}\\
\multicolumn{6}{@{}l@{}}{\footnotesize
  C4 full evasion at 2.2$\times$ supports \Cref{thm:evasion}
  (Evasion-Alignment) in our evaluated setting.}\\
\multicolumn{6}{@{}l@{}}{\footnotesize
  C1 partial evasion is outside the full-evasion scope of
  \Cref{thm:evasion}.}
\end{tabular}
\end{table*}

%% file: sections/08_adaptive.tex

\section{\heading{Adaptive Adversary Analysis}}
\label{sec:adaptive}

Table~\ref{tab:adaptive} reports the adaptive variants.
C2: High-entropy filler and C3: Structured-diverse filler are
included as negative controls: both produced ${\leq}1.5\times$
amplification in smoke testing ($N{=}3$) and were not scaled.
The analysis below focuses on C1 and C4, the two scaled variants.

\subsection{\heading{C1: Topical CSP Weakens Semantic Monitors}}
\label{sec:c1_analysis}

C1 uses vocabulary drawn from the user's question domain as
the decoy, so generated reasoning stays semantically closer to
the query than unrelated Sudoku symbols.
C1 raises query similarity ($a_t$), reducing the magnitude of TP
collapse even when recurrence remains present.
The result is $\approx$50\% joint miss at $11.9\times$
amplification (RG TPR: 21\%, QDM TPR: 38\%).

C1 is nonetheless detectable by length-based defences: B1
catches 94\% and B4 catches 99\% of C1 attacks because the
output remains long regardless of semantic adaptation.
Topical diversification does not reduce output token count.
However, B1 and B4 cannot substitute for semantic monitors
in mixed-workload deployments: their 69\% and 80\% FPR on
HumanEval code generation make them unsuitable as standalone
defences.

C1 is not explained by \Cref{thm:evasion}, which applies
to full TP evasion.
C1 achieves partial evasion: the TP alarm fires on some chunks
but not enough to sustain $S{=}3$ consecutive alarms on most
samples.
It occupies the open middle regime between full evasion and
full detection.

\subsection{\heading{C4: Task-Aligned Filler Requires Low Amplification}}
\label{sec:c4_analysis}

C4 keeps every reasoning chunk task-aligned, evading both
\rg{} and \qdm{} across all 100 samples (0\% TPR each).
The amplification collapses from $22.8\times$ (OverThink) to
$2.2\times$, a 90\% reduction.

This is the experimental support for \Cref{thm:evasion}.
To avoid triggering the TP alarm, the attack must satisfy
$a_t \geq r_t + \tauP$ at every chunk.
In practice, this forces the reasoning to stay close to the
user's question, which causes the model to converge on an
answer rather than loop indefinitely.
C4 supports the theorem's prediction in our evaluated setting:
in our C4 construction, full semantic evasion coincides with
sharply reduced amplification.

C4 is not harmless.
A consistent $2.2\times$ cost overhead is still elevated
API spend.
But it changes the risk profile from a severe denial-of-wallet
threat to a low-amplification slowdown, and cost details are
in \Cref{sec:deployment}.

\subsection{\heading{Security Boundary}}
\label{sec:security_boundary_adaptive}

The adaptive boundary is \emph{amplification under evasion}.

The evaluated high-amplification decoy attacks are detectable.
Full semantic evasion is possible, but in our C4 construction
it reduces the attack to a low-amplification slowdown.
C1 occupies the open middle regime: partial evasion at moderate
amplification.

The failure-mode taxonomy across all defences is:
\begin{itemize}
\item \textbf{B1/B4}: effective against high-amplification
  attacks; unsuitable as standalone monitors on long-form tasks.
\item \textbf{B2}: effective against repetitive attacks;
  destroyed by semantic diversification.
\item \textbf{B3}: catches attacks that produce no answer;
  misses attacks that produce a short one.
\item \textbf{\qdm{}}: effective when output is off-topic;
  weakened by topical alignment and unsuitable for summarisation.
\item \textbf{\rg{}}: effective against recurrence collapse;
  weakened by task-conditioned alignment (C1 open boundary).
\end{itemize}

No single defence works across every case.
\rg{} has the best clean-task FPR profile in our evaluation.
C1 remains the practical open boundary.

%% file: sections/09_deployment.tex

\section{\heading{Deployment Analysis}}
\label{sec:deployment}

\subsection{\heading{Input-Side Safety Classification: Llama Guard~3-1B-INT4}}

Llama Guard~3-1B-INT4~\cite{llamaguard3} ($N{=}50$ attack, $N{=}50$ clean)
detected 0/50 OverThink attack prompts
(TPR~=~0\%, 95\% CI: [0\%, 7.1\%]) while correctly passing all
50 clean prompts (FPR~=~0\%).
The attack prompt contains a legitimate user question alongside a
logic puzzle; neither element triggers a content-safety category.
This result shows that content-safety classification is insufficient
as a standalone defence for the evaluated OverThink prompts.
The harm emerges from the computational consequence of the model's
response, so runtime output-side or reasoning-side monitoring is
needed.

\subsection{\heading{Denial-of-Wallet Cost and Early Stopping Value}}

On Claude Opus~4.7 (Anthropic, May~2026 pricing: \$5/M input
tokens, \$25/M output tokens~\cite{anthropic2026pricing}), the mean attack cost per query is
${\approx}$\$0.402 versus ${\approx}$\$0.00234 for clean queries,
a ${\approx}172\times$ cost amplification.
The $676\times$ value in Table~\ref{tab:closed} is output-token
amplification; the ${\approx}172\times$ value here is total API
cost amplification after applying input and output token pricing.
At 10,000 attacks per day, this corresponds to roughly \$4,000
in API fees compared with roughly \$23 for normal traffic.

\rg{} enables early stopping when non-empty reasoning traces are
available and generation can be cancelled by the client.
On DS-R1-Qwen-7B, \rg{} halts detected OverThink attacks before
the full ${\approx}$8,155-token attacked generation completes,
saving about 92\% of output tokens per detected attack.
This cost-saving result should not be interpreted as applying to
Opus~4.7, where thinking blocks were empty and \rg{} was
structurally inapplicable.

\qdm{} does not enable early stopping because it operates post-hoc
on the final output.
For closed-API deployments where only final outputs are visible,
\qdm{} is the output-side monitor in our design.
It detected all Sonnet~4.5 attacks at $N{=}30$ and all Opus~4.7
smoke-test attacks at $N{=}3$, but it cannot reduce generation cost
because the full output has already been produced.

\subsection{\heading{Deployment Summary}}

Table~\ref{tab:closed} and \Cref{sec:results_closed} report the
closed-API detection statistics.
\rg{} is appropriate when non-empty reasoning traces are accessible
and online cancellation is available; in that setting it provides
runtime detection, early stopping, and near-zero FPR across the clean
task types evaluated here.
\qdm{} is appropriate when only final output text is visible.
Token-budget and length-anomaly baselines should not be used as
standalone defences in mixed-workload deployments because their FPR
on code generation is 69\% and 80\%.
Query-anchored \qdm{} should not be used for summarisation workloads
unless the query anchor is replaced with the source document.

%% file: sections/10_related.tex

\section{\heading{Related Work}}
\label{sec:related}

\subsection{\heading{Reasoning-Chain Resource Attacks}}
\label{sec:rel_attacks}

Resource-exhaustion attacks predate reasoning LLMs.
Sponge examples showed that inputs can be crafted to increase
latency and energy use in neural networks~\cite{shumailov2021sponge},
and OWASP now treats unbounded LLM consumption as a security risk
category~\cite{owasp2024llm}.
Reasoning models make the same risk easier to trigger and easier
to bill: the attacker no longer needs to change model weights or
server code, only the text that reaches the reasoning context.

OverThink~\cite{kumar2025overthink} and
ExtendAttack~\cite{extendattack2025} are the closest attack papers
to this work.
OverThink injects benign-looking decoy reasoning tasks into the
context.
ExtendAttack obfuscates the user's own question so the model spends
tokens reconstructing it before answering.
We use both as primary evaluation targets because they stress
different mechanisms: unrelated constraint solving and sequential
decoding.
Our contribution is not another attack construction.
It is a detector for exposed reasoning traces, plus an adaptive
analysis of where semantic monitors start to fail.

Recent and concurrent work broadens the attack side of the problem.
Excessive Reasoning Attack optimizes inputs that stretch reasoning
without breaking utility~\cite{si2025excessive}.
ReasoningBomb formalizes prompt-induced inference-time DoS for
large reasoning models~\cite{liu2026reasoningbomb}.
ThinkTrap studies black-box prompts that induce very long or
non-terminating thinking on closed services~\cite{li2025thinktrap}.
POT removes the need for poisoned external documents by optimizing
prompt-only overthinking attacks~\cite{li2025pot}.
BadThink and BadReasoner study triggered or backdoored overthinking
behaviour~\cite{liu2025badthink,yi2025badreasoner}.
RECUR, CODE, and Sponge Tool Attack extend the resource-exhaustion
view to reflection-driven reasoning, retrieval-augmented generation, and
tool-augmented agentic reasoning~\cite{wang2026recur,zhang2026code,li2026spongetool}.
These papers sharpen the threat landscape.
\rg{} addresses the defence side for one concrete deployment
condition: non-empty reasoning traces that can be monitored during
generation.

\subsection{\heading{Prompt Injection and Adversarial Prompting}}
\label{sec:rel_injection}

Prompt injection attacks place instructions inside text the model
should treat as data~\cite{perez2022ignore}.
Indirect prompt injection moves the payload into retrieved
documents, web pages, emails, or tool outputs consumed by an
LLM-integrated application~\cite{greshake2023injection}.
This is the natural delivery channel for reasoning-chain
consumption: the injected decoy can sit inside retrieved context
while the user asks a normal question.

Adversarial prompting work has also shown that small text patterns
can steer language-model behaviour across examples or
models~\cite{wallace2019universal,zou2023universal,carlini2023llm}.
Broader coercion attacks can induce misdirection, model control,
denial-of-service, or data extraction~\cite{geiping2024coercing}.
Baseline defences for aligned LLMs often rely on input rewriting,
filtering, perplexity checks, or adversarial training~\cite{jain2023baseline}.
Those defences target harmful outputs or instruction-following
failures.
Reasoning-chain consumption is different.
The prompt can look harmless, and the final output can contain no
policy-violating content.
The harm appears in the trajectory: the model spends its reasoning
budget away from the user's task.
\rg{} therefore monitors the generation path itself rather than
classifying the input alone.

\subsection{\heading{Efficient Reasoning and Early Stopping}}
\label{sec:rel_efficiency}

Several papers study excessive reasoning as an efficiency problem.
Answer-convergence monitoring stops generation when partial answers
stabilize~\cite{liu2025convergence}.
SelfBudgeter learns to allocate reasoning budgets before or during
generation~\cite{selfbudgeter2025}.
THINK-Bench evaluates thinking efficiency and chain-of-thought
quality~\cite{li2025thinkbench}, and recent surveys organize the
larger efficient-reasoning literature~\cite{sui2025stopoverthinking}.

These methods mostly ask when a benign model has reasoned enough.
\rg{} asks a narrower security question: whether the reasoning
trace has left the user's task and collapsed into a self-recurrent
decoy region.
That distinction affects the design.
Benign long-form code or math can be lengthy without being an
attack, so a pure budget policy is too blunt.
Conversely, a short attack that hides reasoning in a closed API
cannot be stopped by a trace monitor.
The deployment section separates these cases instead of treating
all overthinking as the same failure.

\subsection{\heading{Semantic Drift and Runtime Monitoring}}
\label{sec:rel_monitoring}

Embedding-based drift detection is a natural fit for prompt
injection because it can compare semantic representations without
model internals.
ZEDD measures embedding drift between benign and suspect prompt
variants for zero-shot prompt-injection detection~\cite{zedd2025}.
\qdm{} uses a simpler output-side version of the same instinct:
it compares final output chunks against the user query and flags
large drift.
This is useful when reasoning text is hidden, but it has two
limits.
It is post-hoc, so it cannot save tokens, and it can fail on
legitimate summarisation because the answer is about the source
article rather than the short instruction.

\rg{} adds the missing recurrence term.
Instead of asking only whether text is close to the query, it asks
whether each reasoning chunk is closer to the query or to the
model's own previous reasoning.
That task-conditioned comparison is the main difference from
query-output drift monitors and length baselines.
It also defines the adaptive boundary: an attacker can evade the
semantic monitor by keeping reasoning task-aligned, but in our C4
evaluation that reduces amplification to $2.2\times$.

%% file: sections/11_conclusion.tex

\section{\heading{Discussion, Limitations, and Conclusion}}
\label{sec:conclusion}

\subsection{\heading{Discussion}}
\label{sec:discussion}

Two variables decide whether this defence is useful in practice:
how much the attack amplifies generation, and what the defender can
see.
When non-empty reasoning traces are exposed, \rg{} can watch the
generation as it happens and stop a decoy loop before the model
spends the full budget.
When traces are hidden or compressed by the provider, that online
signal is either missing or weaker.
In that setting, \qdm{} is the safer fallback, but it can only act
after the output has already been produced.

The adaptive runs expose the boundary.
Full semantic evasion is possible, but in C4 the amplification falls
to $2.2\times$.
C1 is harder: it keeps $11.9\times$ amplification while pulling the
decoy closer to the user query, which weakens both semantic monitors.

Length rules are useful but not enough.
They catch high-amplification attacks, including C1, but they also
flag many clean HumanEval code answers.
That makes them poor standalone monitors for mixed workloads.
\rg{} performs better on the clean tasks tested here because it
does not treat length or query distance alone as evidence of attack;
it looks for task drift and recurrence together.

\subsection{\heading{Limitations}}
\label{sec:limitations}

\noindent\pointlabel{Reasoning visibility.}
\rg{} requires non-empty reasoning text.
It is structurally inapplicable when an API hides reasoning traces,
as observed in the Opus~4.7 smoke test.
Provider-summarised thinking, as observed for Sonnet~4.5, weakens
chunk-level recurrence signals and reduces \rg{} performance.
This is a deployment boundary, not a calibration issue.

\noindent\pointlabel{Adaptive coverage.}
C1 and C4 are evaluated at $N{=}100$ only on DS-R1-Qwen-7B.
The results identify an adaptive boundary on the primary model,
but they do not establish transferability to all reasoning models.
Testing the same stress tests across more open-weight and closed
models remains future work.

\noindent\pointlabel{C1: Topical CSP remains unsolved.}
C1 shows that topical alignment can preserve meaningful
amplification while weakening semantic monitors.
Length-based baselines catch most C1 samples, but their false
positive rates on long-form clean tasks are too high for standalone
deployment.
A stronger defence for C1 likely needs task-aware anchoring beyond
the short user query, such as source-document anchors in RAG.

\noindent\pointlabel{No white-box adaptive attack.}
We do not evaluate gradient-based optimisation against the sentence
encoder or the \rg{} decision rule.
Such an attacker requires access to defence internals and is outside
the black-box prompt-injection threat model used in this paper.

\noindent\pointlabel{Proxy baselines.}
The compression-ratio baseline is a proxy for repetitiveness, not
true perplexity.
We use it because saved text outputs and closed APIs do not expose
token-level logits.
Its adaptive-attack results should therefore be read as proxy
baseline results, not as logit-based perplexity measurements.

\subsection{\heading{Conclusion}}
\label{sec:conclusion_final}

\rg{} detects that failure mode when the reasoning trace is visible.
It looks for chunks that stop moving toward the user task and start
returning to the model's own recent reasoning.
On the primary open-weight model evaluation, \rg{} detects 99\% of
OverThink attacks and 92\% of ExtendAttack instances, while keeping
false positives near zero on the clean task types we tested.
Because the detector runs during generation, it can also stop
detected attacks early and save about 92\% of output tokens.

The limit is equally clear.
\rg{} needs non-empty reasoning traces.
If reasoning is summarised, \rg{} can run with reduced sensitivity;
if reasoning is hidden, it is structurally inapplicable.
\qdm{} remains the post-hoc fallback in both cases, but it cannot
save the generation cost.
The next useful tests are end-to-end RAG deployments, adaptive runs
on more models, document-anchored semantic monitors, and
entropy-based signals when token-level logits are available.